\documentclass[11pt]{scrartcl}
\usepackage[a4paper, total={16cm, 24cm}]{geometry}

\usepackage{multicol}
\usepackage{tocloft}

\usepackage[small]{caption}
\usepackage{subcaption}

\usepackage{algorithm}
\usepackage[table,dvipsnames]{xcolor}
\usepackage{todonotes}
\usepackage{pdflscape}
\usepackage{adjustbox}

\usepackage{natbib}
\usepackage[pagebackref]{hyperref}
\hypersetup{
	pdfencoding=auto,
	psdextra,
	colorlinks=true,
	citecolor=green!50!black,
	linkcolor=red!60!black,
	urlcolor=purple!70!black
}

\renewcommand{\epsilon}{\varepsilon}

\usepackage{authblk}

\usepackage{url}
\usepackage{microtype}
\usepackage[utf8]{inputenc}
\usepackage{graphicx}
\usepackage{amsmath}
\usepackage{amsthm}
\usepackage{amssymb}
\usepackage{mathtools}
\usepackage{nicefrac}
\usepackage{booktabs}
\usepackage{enumitem}
\usepackage{tabulary}
\usepackage{tabularx}
\usepackage{multirow}

\usepackage{wrapfig}

\usepackage[nameinlink]{cleveref}
\usepackage{stmaryrd}
\usepackage{thm-restate}
\usepackage{tikz}
\usetikzlibrary{arrows.meta}
\usepackage{tcolorbox}

\newtheoremstyle{sfthm}
	{\topsep}
	{\topsep}
	{\itshape}
	{}
	{\sffamily\bfseries}
	{}
	{.5em}
	{}
\theoremstyle{sfthm}
\newtheorem{theorem}{Theorem}[section]
\newtheorem{lemma}[theorem]{Lemma}
\newtheorem{proposition}[theorem]{Proposition}
\newtheorem{corollary}[theorem]{Corollary}

\newtheoremstyle{sfdef}
{\topsep}
{\topsep}
{}
{}
{\sffamily\bfseries}
{}
{.5em}
{}
\theoremstyle{sfdef}

\newcommand{\chiStable}{\chi_{\mathrm{stable}}}

\newcommand{\tw}{\operatorname{tw}}

\title{Graph Coloring with Color Preferences}

\author{Tomohiro Koana}
\author{Yeeseok Oh}
\author{Hirotaka Yoneda}
\affil{The University of Tokyo}
\date{\vspace{-3em}}

\begin{document}

\maketitle

\begin{abstract}
\begin{center}
	\textbf{\textsf{Abstract}} \smallskip
\end{center}
We study graph coloring with color preferences, in which each vertex ranks the available colors. In addition to assigning different colors to adjacent vertices, we require the coloring to be \emph{stable}: no group of vertices can cyclically exchange their assigned colors so that each strictly prefers its new color to its original one. We define the \emph{stable chromatic number} $\chiStable(G)$ of a graph $G$ as the minimum integer $k$ such that every preference profile admits a stable $k$-coloring of $G$.
We establish several upper and lower bounds. In particular, for any acyclic orientation of the edges of $G$, the largest number of vertices reachable from a vertex by directed paths, including the vertex itself, is an upper bound on $\chiStable(G)$. This shows that $\chiStable(G)$ is well-defined. We also show that $O(t \log (1+n/t))$ colors suffice for an $n$-vertex graph $G$ of treewidth $t$, and complement this with a lower bound in terms of the Grundy number.
Turning to the problem of finding a minimum stable coloring for a given profile, we show that stable $2$-colorability is polynomial-time solvable, whereas stable $k$-colorability is NP-complete for every fixed $k\ge 3$. Using the treewidth bound, we give a fixed-parameter tractable algorithm parameterized by treewidth.
\end{abstract}

\section{Introduction}
Graph coloring models the problem of assigning reusable resources under pairwise incompatibilities: vertices are recipients, colors are resource types, and adjacent vertices must receive different colors \citep{jensen1995graph}. Classical graph coloring does not distinguish between colors and asks only for feasibility. In many allocation settings, however, vertices represent agents who have strict preferences over the available colors.
We study this preference-aware setting under a new stability requirement: there is no cycle along which every agent prefers the color assigned to the next agent.

To illustrate, consider designing a take-home exam for a course. You announce that the exam is to be solved individually, but you are aware that some students close to each other might work on it together. You therefore decide to make multiple versions of the exam so that close friends receive different versions. To minimize the workload, you want to make as few versions as possible.
When students are modeled as vertices, two students are adjacent if they might work together. A proper coloring of this graph represents a feasible assignment of exam versions.

Then a problem arises because students may rank the versions differently based on their strengths in different topics, familiarity with the material, or perceptions of the versions' difficulty. Once the packets have been distributed, students connected through the collaboration network may also exchange them privately.
For example, \Cref{fig:exam-before-after} shows such a situation: students on a cycle may circulate their packets so that each receives the version initially held by the next student. If every participant strictly prefers the received version, this is an individually profitable coalition deviation.

Such a deviation may undermine the instructor's assignment. In particular, a student may acquire the same version as a collaborator outside the deviating coalition, thereby violating properness. Even when properness is preserved, an unauthorized exchange can undermine the intended distribution of assessment variants. We call a directed cycle along which every student prefers the version assigned to the next student a \emph{blocking cycle}, and we seek a proper assignment with no such cycle. We call such a coloring \emph{stable}. Preparing and validating many distinct versions is costly, which motivates minimizing the number of colors.

In this paper, we focus on the following two basic questions: How many colors suffice to stably color a graph? What is the complexity of finding such a coloring?

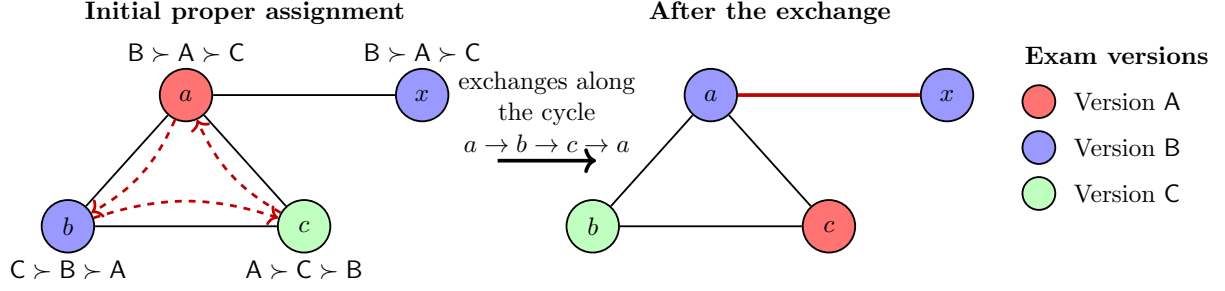
\begin{figure*}[t]
\centering
\resizebox{\textwidth}{!}{%
\begin{tikzpicture}[
    conflict/.style={draw=black, thick},
    envy/.style={->, dashed, draw=red!75!black, very thick},
    student/.style={
        circle,
        draw=black,
        thick,
        minimum size=8mm,
        inner sep=0pt,
        font=\normalsize
    },
    versionA/.style={student, fill=red!55},
    versionB/.style={student, fill=blue!40},
    versionC/.style={student, fill=green!25},
    pref/.style={
        font=\normalsize,
        fill=white,
        inner sep=1pt,
        text depth=0pt
    },
    paneltitle/.style={font=\normalsize\bfseries}
]

\begin{scope}
    \node[paneltitle] at (0.9,3.25)
        {Initial proper assignment};

    \node[versionA] (a1) at (0,2.0) {$a$};
    \node[versionB] (b1) at (-1.8,0) {$b$};
    \node[versionC] (c1) at (1.8,0) {$c$};
    \node[versionB] (x1) at (3.6,2.0) {$x$};

    \draw[conflict] (a1) -- (b1);
    \draw[conflict] (b1) -- (c1);
    \draw[conflict] (c1) -- (a1);
    \draw[conflict] (a1) -- (x1);

    \node[pref] at (0,2.65)
        {$\mathsf{B}\succ\mathsf{A}\succ\mathsf{C}$};

    \node[pref] at (-1.8,-0.65)
        {$\mathsf{C}\succ\mathsf{B}\succ\mathsf{A}$};

    \node[pref] at (1.8,-0.65)
        {$\mathsf{A}\succ\mathsf{C}\succ\mathsf{B}$};

    \node[pref] at (3.6,2.65)
        {$\mathsf{B}\succ\mathsf{A}\succ\mathsf{C}$};

    \draw[envy, bend left=18] (a1) to (b1);
    \draw[envy, bend left=18] (b1) to (c1);
    \draw[envy, bend left=18] (c1) to (a1);

\end{scope}

\draw[->, ultra thick] (4.75,1.0) -- (6.25,1.0);

\node[font=\normalsize, align=center] at (5.5,1.75)
    {exchanges along\\
     the cycle\\
     $a\to b\to c\to a$};

\begin{scope}[xshift=8cm]
    \node[paneltitle] at (0.9,3.25)
        {After the exchange};

    \node[versionB] (a2) at (0,2.0) {$a$};
    \node[versionC] (b2) at (-1.8,0) {$b$};
    \node[versionA] (c2) at (1.8,0) {$c$};
    \node[versionB] (x2) at (3.6,2.0) {$x$};

    \draw[conflict] (a2) -- (b2);
    \draw[conflict] (b2) -- (c2);
    \draw[conflict] (c2) -- (a2);

    \draw[draw=red!75!black, ultra thick] (a2) -- (x2);

\end{scope}

\begin{scope}[xshift=13.0cm, yshift=1.0cm]
    \node[font=\normalsize\bfseries, anchor=west] at (-0.35,1.6)
        {Exam versions};

    \node[versionA, minimum size=5mm] at (0,0.9) {};
    \node[font=\normalsize, anchor=west] at (0.4,0.9)
        {Version $\mathsf{A}$};

    \node[versionB, minimum size=5mm] at (0,0.2) {};
    \node[font=\normalsize, anchor=west] at (0.4,0.2)
        {Version $\mathsf{B}$};

    \node[versionC, minimum size=5mm] at (0,-0.5) {};
    \node[font=\normalsize, anchor=west] at (0.4,-0.5)
        {Version $\mathsf{C}$};
\end{scope}
\end{tikzpicture}%
}
\caption{
Exchanging versions along the blocking cycle $a\to b\to c\to a$ benefits all participants but makes the initially proper assignment improper.
}
\label{fig:exam-before-after}
\end{figure*}

\paragraph{Our Contributions}
We introduce the \emph{stable chromatic number} $\chiStable(G)$ of a graph $G$, the minimum integer $k$ such that every preference profile on $G$ admits a stable $k$-coloring. In \Cref{sec:bounds}, we bound this invariant from above and below.
For any acyclic orientation of the edges of $G$, the largest number of vertices reachable from a vertex by directed paths, including the vertex itself, is an upper bound on $\chiStable(G)$. In particular, $\chiStable(G)\le |V(G)|$, so the stable chromatic number is well-defined. For the lower bound, we use the Grundy number $\Gamma(G)$, which is the maximum number of colors that First-Fit can use over all
vertex orderings. From a Grundy coloring with $r$ colors, we construct a preference profile that admits no stable $(r-1)$-coloring, proving that $\Gamma(G)\le\chiStable(G)$.
Since every stable coloring is proper, $\chi(G)$ is a natural lower bound on $\chiStable(G)$. We therefore focus mainly on graph classes of bounded chromatic number and ask whether their stable chromatic number is also bounded. We determine the stable chromatic number exactly for complete bipartite graphs, paths, and cycles. \Cref{tab:bounds-summary} summarizes
these bounds.

In \Cref{sec:algorithms}, we turn to the complexity of finding a stable coloring with the minimum number of colors for a graph $G$ and a preference profile $\succ$. We show that stable $2$-colorability is polynomial-time solvable, while stable $k$-colorability is NP-complete for every fixed $k\ge 3$.
In view of this NP-hardness, we seek fixed-parameter tractable algorithms.
An algorithm is \emph{fixed-parameter tractable} (FPT) with respect to a parameter $t$ if it runs in time $f(t)N^{O(1)}$ on instances of size $N$, for some computable function $f$.
We focus on treewidth, arguably the most extensively studied structural graph parameter, which measures how closely a graph resembles a tree.
We show that $O(t\log n)$ colors suffice for every $n$-vertex graph of treewidth $t$. Using this bound, we design a fixed-parameter tractable algorithm parameterized by treewidth. In particular, a stable coloring of $(G,\succ)$ using the minimum possible number of colors can be found in polynomial time when $G$ is a tree.

\begin{table*}[t]
    \centering
    \begin{tabularx}{\textwidth}{@{}
        >{\raggedright\arraybackslash}p{0.38\textwidth}
        >{\raggedright\arraybackslash}p{0.12\textwidth}
        >{\raggedright\arraybackslash}X
        >{\raggedright\arraybackslash}X
        @{}}
        \toprule
        \multirow{2}{*}{Graph Class} & \multirow{2}{*}{$\chi(G)$}
            & \multicolumn{2}{c}{Stable Chromatic Number $\chiStable(G)$} \\
        \cmidrule(lr){3-4}
            & & Upper Bound & Lower Bound \\
        \midrule
        General $n$-vertex graphs
            & $\le n$
            & $n$
            & $\Gamma(G)$ \\
        Paths $P_n$ ($n \ge 4$)
            & $2$
            & $3$
            & $3$ \\
        Cycles $C_n$
            & $\le 3$
            & $3$
            & $3$ \\
        Complete bipartite graphs $K_{m,n}$
            & $2$
            & $\min\{m,n\}+1$
            & $\min\{m,n\}+1$ \\
        Graphs of maximum degree $\Delta$
            & $\le\Delta+1$
            & $2^\Delta$
            & $\Gamma(G)$ \\
        Trees
            & $\le 2$
            & $O(\log n)$
            & $\Omega(\log n)$ \\
        Graphs of treewidth $t$
            & $\le t+1$
            & $O(t\log(1+n/t))$
            & $\Omega(t\log(1+n/t))$ \\
        Planar graphs
            & $\le 4$
            & $O(\log^2 n)$
            & $\Omega(\log n)$ \\
        \bottomrule
    \end{tabularx}
        \caption{Bounds on the chromatic number $\chi(G)$ and the stable chromatic number $\chiStable(G)$. Here $\Delta(G)$ and $\Gamma(G)$ denote the maximum degree and Grundy number, respectively.}
    \label{tab:bounds-summary}
\end{table*}

\paragraph{Related Work}
The closest related work is by \citet{samaris2023}, who studies graph coloring, possibly with color lists, in which vertices have linear preferences over the available colors. A coloring is \emph{Pareto optimal} if no other proper coloring makes some vertex better off without making any vertex worse off. \citeauthor{samaris2023} characterizes Pareto-optimal colorings and gives a mechanism for finding one.

Pareto optimality is incomparable with our notion of stability. A non-Pareto-optimal coloring can be stable because a Pareto improvement need not be realizable by a cyclic trade. Conversely, a Pareto-optimal coloring can be unstable because rotating colors along a blocking cycle need not produce another proper coloring.

Several other works also consider preferences or stability in graph coloring.
\citet{blum2008multiagent} study a multiagent
variant of graph coloring in which each agent controls a subgraph and ranks the colors used on that subgraph, with objectives including Pareto efficiency, fairness, and individual rationality.
\citet{kun2013anti} also use the term \emph{stable coloring} to mean a pure Nash equilibrium in an anti-coordination game. Unlike our model, their colorings need not be proper, vertices have no intrinsic color preferences, and stability concerns unilateral recoloring rather than envy cycles.

Beyond graph coloring, stability has been studied extensively in resource-allocation problems. In the house-allocation model of \citet{shapley1974cores}, a Pareto-optimal, core-stable allocation can be found by eliminating \emph{trading cycles}. Such a trading cycle corresponds exactly to a blocking cycle under a proper coloring of the complete graph $K_n$. Thus, our model of stable coloring generalizes the house-allocation model.
Stability in terms of trading has also been considered in matchings \citep{alcalde1994,chenEtAl2021}.

\section{Preliminaries}

For a nonnegative integer $t$, let $[t]=\{1,2,\ldots,t\}$.
Throughout the paper, a graph $G$ is finite, simple, undirected, and nonempty. We write $V$ and $E$ for its vertex and edge sets when no confusion arises, and let $n=|V|$. For a vertex $v\in V$, let $N(v)$ denote the set of neighbors of $v$, and let $d(v)=|N(v)|$ be its degree.
The maximum degree of $G$ is denoted by $\Delta(G)=\max_{v \in V}d(v)$.
We use $\tw(G)$ to denote the treewidth of $G$, defined in \Cref{subsec:tw}. 

A graph $G$ is \emph{bipartite} if its vertex set can be partitioned into two sets such that every edge has one endpoint in each set. We write $K_{m,n}$ for the complete bipartite graph with partite sets of sizes $m$ and $n$. We write $P_n$ and $C_n$ for the path and cycle on $n$ vertices, respectively.

A \emph{proper coloring}, or simply a \emph{coloring}, of $G$ is a function $\phi:V \to \mathbb{N}$ such that $\phi(u)\ne \phi(v)$ for every edge $uv\in E$. A coloring is called a \emph{$k$-coloring} if $\phi(V) \subseteq [k]$. The \emph{chromatic number} of $G$, denoted $\chi(G)$, is the minimum $k$ such that $G$ admits a $k$-coloring.

A coloring $\gamma:V\to [r]$ is called a \emph{Grundy coloring} if every vertex $v$ with $\gamma(v)=i$ has, for each $j<i$, a neighbor $u$ with $\gamma(u)=j$. The \emph{Grundy number} of $G$, denoted by $\Gamma(G)$, is the largest integer $r$ for which $G$ admits a Grundy coloring $\gamma$ with $\gamma(V)=[r]$ \citep{grundy1939mathematics,christen1979perfect}.
Equivalently, $\Gamma(G)$ is the maximum number of colors that can be used by the First-Fit greedy coloring algorithm over all orderings of $V$, where each vertex is colored with the smallest positive integer that does not appear on its previously colored neighbors.

A \emph{stable coloring instance} $I=(G,\succ)$ consists of an undirected graph $G=(V,E)$ and, for each $v\in V$, a strict total order $\succ_v$ on $\mathbb{N}=\{1,2,\ldots\}$. We write
$\succ=(\succ_v)_{v\in V}$ for the resulting preference profile.

Only a finite prefix of each ranking is relevant to the algorithms in this paper. For a prescribed number $k$ of colors, we explicitly list each ranking on $[\min\{k,n\}]$. If $k\le n$, this set is $[k]$.
If $k>n$, \Cref{cor:finite} guarantees a stable coloring using colors from $[n]$, so preferences among colors larger than $n$ do not affect the answer or the construction. When minimizing the number of colors, we list each ranking on $[n]$. We preprocess these lists to support constant-time preference comparisons. All running times include this preprocessing and are measured in the combined size of the graph and the lists.

Given an instance $I=(G,\succ)$ and a coloring $\phi$ of $G$, the \emph{envy graph} of $\phi$ is the directed graph $D_\phi=(V,A_\phi)$, where $(u,v)\in A_\phi$ if and only if $uv\in E$ and $\phi(v)\succ_u\phi(u)$.
A \emph{blocking cycle} is a directed cycle in $D_\phi$. Equivalently, for some $t\ge 2$, it is a sequence $v_1v_2\ldots v_tv_1$ such that $v_1,\ldots,v_t$ are distinct, $v_iv_{i+1}\in E$, and $\phi(v_{i+1})\succ_{v_i}\phi(v_i)$ for every $i\in[t]$, where $v_{t+1}=v_1$.
In particular, if an edge $uv\in E$ satisfies $\phi(u)\succ_v\phi(v)$ and $\phi(v)\succ_u\phi(u)$, then $uvu$ is a blocking cycle. A coloring with no blocking cycle is called a \emph{stable coloring}. Thus, a coloring is stable if and only if its envy graph is acyclic, or equivalently, if the vertices admit an ordering in which every envy arc goes from an earlier vertex to a later vertex.

As seen in \Cref{fig:exam-before-after}, a blocking cycle need not support a feasible simultaneous recoloring.
Assigning each $v_i$ the color $\phi(v_{i+1})$ preserves properness on the edges of the cycle. The resulting coloring is proper if and only if no conflict arises on a chord of the cycle or on an edge between the cycle and $V\setminus\{v_1,\ldots,v_t\}$.

Given a graph $G$ and a preference profile $\succ$, we define $\chiStable(G, \succ)$ as the minimum integer $k$ such that there exists a stable $k$-coloring. We define the \emph{stable chromatic number} of $G$ as
$$\chiStable(G) = \sup_{\succ} \chiStable(G, \succ),$$
where the supremum is taken over all preference profiles $\succ$ on $\mathbb{N}$.

\section{Bounds on the Stable Chromatic Number}\label{sec:bounds}
In this section, we investigate the number of colors we need to stably color a graph.
The results are organized into three parts. \Cref{subsec:upper,subsec:lower,subsec:monotone} establish our general tools: an upper bound based on reachability in acyclic orientations, a lower bound based on the Grundy number, and monotonicity under taking subgraphs. 
\Cref{subsec:pathcyclebipart} applies these tools to determine exact values for several elementary graph classes. \Cref{subsec:maxdegree,subsec:tw} derive structural bounds in terms of maximum degree and treewidth. We further give sharper bounds for planar graphs in \Cref{subsec:planar}.

\subsection{Upper Bound via Reachability}\label{subsec:upper}

Our main tool for upper bounds is the following reachability parameter.
Many upper bounds in the subsequent subsections are obtained by choosing
suitable acyclic orientations and applying \Cref{thm:reachability}. The
theorem also implies $\chiStable(G)\le n$ (\Cref{cor:finite}), and hence
that the stable chromatic number is well-defined.

    Given a graph $G$, define $\mathrm{reach}(G)$ as the minimum, over all acyclic edge orientations $D$ of $G$, of
    \begin{equation*}
        \max_{v \in V} \{\text{number of vertices reachable from $v$ (including $v$)}\}.\footnote{Acyclic orientations and reachability were previously studied by \citet{even2018best} in the context of centralized local algorithms.}
    \end{equation*}
%
We show that $\mathrm{reach}(G)$ is an upper bound on the stable
chromatic number. The proof colors the vertices in reverse topological order
of an orientation attaining $\mathrm{reach}(G)$, assigning each vertex its
most preferred color not used by any other vertex reachable from it, so every
envy arc is directed consistently with the orientation.

\begin{theorem} \label{thm:reachability}
    $\chiStable(G) \leq \mathrm{reach}(G)$.
\end{theorem}

\begin{proof}
    Let $D$ be an acyclic edge orientation that attains
    $\mathrm{reach}(G)$. For each $v\in V$, let $S_v$ be the set of
    vertices other than $v$ that are reachable from $v$ in $D$. Process
    the vertices in reverse topological order. When processing $v$, let
    $\varphi(v)$ be its most preferred color in
    $[\mathrm{reach}(G)]\setminus\{\varphi(x):x\in S_v\}$.

    Every vertex in $S_v$ occurs after $v$ in the topological order and
    is therefore colored before $v$. Since
    $|S_v|\le\mathrm{reach}(G)-1$, an available color always exists.
    Moreover, every out-neighbor of $v$ belongs to $S_v$, so the
    resulting coloring is proper.

    To prove stability, consider an edge $v\to u$ of $D$. Vertex $u$ is
    colored before $v$. We claim that
    $\varphi(u)\succ_u\varphi(v)$. Suppose instead that
    $\varphi(v)\succ_u\varphi(u)$. Since $u$ chose its most preferred
    available color, $\varphi(v)$ must belong to
    $\{\varphi(x):x\in S_u\}$. The edge $v\to u$ implies
    $S_u\subseteq S_v$, so $\varphi(v)$ was excluded when $v$ was
    colored, a contradiction. Thus, an envy arc on this edge can only
    be directed from $v$ to $u$. The envy graph is therefore a subgraph
    of $D$ and is acyclic.
\end{proof}

\begin{corollary} \label{cor:finite}
    $\chiStable(G) \leq n$.
\end{corollary}

\begin{proof}
    By definition, $\mathrm{reach}(G)$ is at most $n$. Therefore, by \Cref{thm:reachability}, $\chiStable(G) \leq n$ holds.
\end{proof}

\subsection{Lower Bound via Grundy Colorings}\label{subsec:lower}

Every stable coloring is proper, so $\chi(G)\le\chiStable(G)$.
We strengthen this lower bound using the Grundy number.

\begin{theorem}\label{thm:grundy-lower}
    $\Gamma(G) \leq \chiStable(G)$.
\end{theorem}

\begin{proof}
Let $r=\Gamma(G)$. If $r=1$, then the claim follows trivially from $1\le\chi(G)\le\chiStable(G)$, so assume that $r\ge 2$. Fix a Grundy coloring $\gamma:V \to [r]$ using all $r$ colors.
Thus, if $\gamma(v)=i$, then for every $j<i$, vertex $v$ has a neighbor $u$ with $\gamma(u)=j$.

We construct a preference profile $\succ$ under which no stable $(r-1)$-coloring exists.
For each $i \in [r-1]$ and $v \in \gamma^{-1}(i)$, define the restriction of $\succ_v$ to $[r-1]$ by
$$
	i\succ_v i+1\succ_v \cdots\succ_v r-1
	\succ_v i-1\succ_v i-2\succ_v\cdots\succ_v 1.
$$
For every vertex $v \in \gamma^{-1}(r)$, define the restriction of $\succ_v$ to $[r-1]$ by
$$
	r-1\succ_v r-2\succ_v\cdots\succ_v 1.
$$
Extend these orders arbitrarily to strict rankings of $\mathbb{N}$.
The key property is that if a vertex $v$ with $\gamma(v)=i$ currently has a color $c<i$, then it strictly prefers every color $d\in[r-1]$ with $d>c$ to $c$.

Suppose for contradiction that there is a stable $(r-1)$-coloring $\phi$.
We prove by induction on $i$ that every vertex $v \in \gamma^{-1}(i)$ receives a color $\phi(v)\ge i$.

For $i=1$, the claim follows because every color is at least $1$.
Assume the result holds for all $j<i$.
Take $v \in \gamma^{-1}(i)$ and suppose $\phi(v)=j<i$.
Since $\gamma$ is a Grundy coloring, $v$ has a neighbor $u$ such that $\gamma(u)=j$.
By the induction hypothesis, $\phi(u)\ge j$.
If $\phi(u)=j$, then $u$ and $v$ have the same color, contradicting properness of $\phi$.
Hence $\phi(u)>j$.

Since $u \in \gamma^{-1}(j)$ and
$j<\phi(u)\le r-1$, we have $j \succ_u \phi(u)$.
On the other hand, since $v \in \gamma^{-1}(i)$ has color $\phi(v)=j<i$ and $\phi(u)>j$, we have $\phi(u) \succ_v j$.
Thus $uvu$ is a blocking cycle, which contradicts stability.

Therefore every vertex $v \in \gamma^{-1}(i)$ must have $\phi(v) \ge i$.
In particular, a vertex $v \in \gamma^{-1}(r)$ would need a color at least $r$, a contradiction.
Hence, $\chiStable(G)\ge r=\Gamma(G)$.
\end{proof}

Hence, for a general graph $G$, we have
\begin{equation*}
    \chi(G) \le \Gamma(G) \le \chiStable(G) \le \mathrm{reach}(G).
\end{equation*}

\subsection{Monotonicity under Subgraphs}\label{subsec:monotone}

Unlike the Grundy number, the stable chromatic number is monotone under
taking subgraphs. This implies that the Grundy lower bound is not tight.

\begin{proposition}\label{prop:tcf-monotone}
The stable chromatic number is monotone under taking subgraphs. That is, if
$H\subseteq G$, then $\chiStable(H)\le \chiStable(G)$.
\end{proposition}
\begin{proof}
Let $k=\chiStable(G)$. Fix an arbitrary preference profile $\succ^H$ on $H$ and extend it arbitrarily to a preference profile $\succ^G$ on $G$. Take a stable $k$-coloring $\phi$ of $(G, \succ^G)$ and restrict it to $V(H)$. This restriction is a proper $k$-coloring of $H$, since every edge of $H$ is also an edge of $G$.
The resulting envy digraph on $H$ is a subdigraph of the envy digraph of $\phi$ on $G$. It is therefore acyclic, so the restriction is a stable $k$-coloring of $H$.
\end{proof}

\subsection{Paths, Cycles, and Complete Bipartite Graphs}\label{subsec:pathcyclebipart}

Paths satisfy $\chi(P_n)\le 2$, cycles satisfy $\chi(C_n)\le 3$, and complete bipartite graphs satisfy $\chi(K_{m,n})=2$ for $m,n\ge 1$.
We determine the stable chromatic numbers of these three graph classes.

\begin{proposition}\label{prop:paths}
    For every path $P_n$ on $n$ vertices,
    \begin{equation*}
        \chiStable(P_n) =
        \begin{cases}
            1 & \text{if} \ \, n=1,\\
            2 & \text{if} \ \, 2 \le n \le 3,\\
            3 & \text{if} \ \, n\ge 4.
        \end{cases}
    \end{equation*}
\end{proposition}

\begin{proof}
    Let $v_1, \dots, v_n$ be the vertices of $P_n$ in path order, so its
    edges are $v_iv_{i+1}$ for $i=1,\dots,n-1$. The lower and upper
    bounds are as follows.
    \begin{itemize}
        \item $\Gamma(P_1) \geq 1$, $\Gamma(P_2) \geq 2$,
        $\Gamma(P_3) \geq 2$, and $\Gamma(P_n) \geq 3$ for $n \geq 4$.
        For the last bound, process $v_1,v_2,v_4,v_3$ first and the
        remaining vertices arbitrarily. First-Fit then assigns color $3$
        to $v_3$.
        \item $\mathrm{reach}(P_1) \leq 1$, $\mathrm{reach}(P_2) \leq 2$, $\mathrm{reach}(P_3) \leq 2$, and $\mathrm{reach}(P_n) \leq 3$ for $n \geq 4$. These bounds are obtained by orienting each edge $v_i v_{i+1}$ in the direction $v_i \to v_{i+1}$ for odd $i$, and $v_{i+1} \to v_i$ for even $i$.
    \end{itemize}
    The displayed lower and upper bounds coincide for every $n$, so
    \Cref{thm:reachability,thm:grundy-lower} give the claimed values.
\end{proof}

\begin{restatable}{proposition}{cycles}\label{prop:cycles}
    For every cycle $C_n$ on $n \geq 3$ vertices, $\chiStable(C_n) = 3$.
\end{restatable}
\begin{proof}
First, we prove the lower bound $\chiStable(C_n) \geq 3$.
\begin{itemize}
    \item $n = 3$: Since $C_3$ is not $2$-colorable, $\chi_{\mathrm{stable}}(C_3) \geq 3$.
    \item $n \geq 4$: $C_n$ contains $P_4$ as a subgraph. \Cref{prop:paths} states that $\chiStable(P_4) = 3$, so by \Cref{prop:tcf-monotone}, $\chiStable(C_n) \geq 3$.
\end{itemize}

Next, we prove the upper bound $\chiStable(C_n) \leq 3$. If $n$ is even, we use the reachability bound: when we orient the edges of $C_n$ in alternating directions, at most three vertices are reachable from each vertex. Therefore, $\mathrm{reach}(C_n) = 3$, and by \Cref{thm:reachability}, $\chiStable(C_n) \leq 3$.

It remains to show that, if $n$ is odd, three colors suffice for $C_n$.
Fix an arbitrary preference profile $\succ$.
For each vertex $v$, let $t(v)$ denote its favorite color among $\{1,2,3\}$.
First, suppose that not all vertices have the same favorite color.
Partition the vertices of the cycle into maximal sequences of consecutive vertices on which $t(v)$ is constant.
Consider one such sequence $R=(x_1,\dots,x_\ell)$, and suppose that all vertices in $R$ have favorite color $a$. Let $d$ be the favorite color of the first vertex in the next run (which exists as not all vertices have the same favorite color). Since the sequences are maximal, $d \ne a$.

Choose an auxiliary color $b_R\neq a$ as follows.
If $\ell$ is odd, choose any color $b_R\in\{1,2,3\}\setminus\{a\}$.
If $\ell$ is even, choose the unique color $b_R\in\{1,2,3\}\setminus\{a,d\}$.

Color the vertices in $R$ alternately by
\[
    \phi(x_j)=
    \begin{cases}
        a, & j\text{ odd},\\
        b_R, & j\text{ even}.
    \end{cases}
\]
Doing this for every sequence gives a proper coloring.
Inside a sequence, adjacent vertices alternate between $a$ and $b_R$.
At a boundary between runs, the first vertex of the next sequence receives its favorite color $d$. The last vertex of the preceding sequence receives $a \ne d$ if that run has odd length and $b_R \ne d$ if it has even length. Hence, the coloring is proper.

Every edge has at least one endpoint receiving its favorite color.
Hence no edge forms a blocking $2$-cycle.
Moreover, a vertex receiving its favorite color has outdegree zero in the envy graph. Any directed cycle in a digraph whose underlying graph is a subgraph of $C_n$ is either a blocking $2$-cycle or traverses the whole cycle.
Neither is possible, so the coloring is stable.

It remains to handle the case in which all vertices have the same
favorite color $a$.
Let the other two colors be $b,c$.
Choose one edge $uv$.
Color $u$ and $v$ with $b,c$ in whichever of the two orders avoids a blocking $2$-cycle on $uv$.
One of the two orders must do so: if the assignment $u=b,v=c$ forms a blocking $2$-cycle, then $u$ prefers $c$ to $b$ and $v$ prefers $b$ to $c$, so the swapped assignment gives both vertices their preferred color among $b,c$.
Now color the remaining path of $n-2$ vertices alternately with $a,b$, starting and ending with $a$. Every edge except possibly $uv$ has an endpoint colored $a$, and $uv$ was chosen not to form a blocking $2$-cycle.
Thus there is no blocking $2$-cycle, and a vertex colored $a$ has outdegree zero. As above, the envy graph is acyclic.
\end{proof}
The case of odd cycles shows that the reachability upper bound is not tight. For odd $n \ge 5$, we claim that $\mathrm{reach}(C_n) \ge 4$. Consider an arbitrary acyclic orientation of $C_n$. Assume there exists an acyclic orientation of $C_n$ where every vertex can reach at most three vertices (including itself). Then, there cannot exist a vertex with one incoming and one outgoing edge: such a vertex is reachable from some source vertex, and the resulting directed path of length at least two, together with the source’s other outgoing neighbor, would give at least four reachable vertices. This implies that sources and sinks must alternate around the cycle, which is impossible when the cycle has odd length. Hence, $\chiStable(C_n) = 3 < 4 \le \mathrm{reach}(C_n)$.

\begin{proposition}\label{prop:complete-bipartite}
Let $m,n$ be integers with $1\le m\le n$. Then
$\chiStable(K_{m,n})=m+1$.
\end{proposition}

\begin{proof}
Let $X=\{x_1,\dots,x_m\}$ and $Y=\{y_1,\dots,y_n\}$ be the two partite sets of $K_{m,n}$. Orient every edge from $Y$ to $X$. This orientation shows that $\mathrm{reach}(K_{m,n})\le m+1$, and hence $\chiStable(K_{m,n})\le m+1$ by \Cref{thm:reachability}.

For the lower bound, let $H$ be the graph obtained from $K_{m,n}$ by
removing $x_i y_i$ for every $i\in[m-1]$. Applying First-Fit in the order
\[
    x_1,y_1,\dots,x_{m-1},y_{m-1},x_m,y_m,y_{m+1},\dots,y_n
\]
uses $m+1$ colors, so $\Gamma(H)\ge m+1$. By
\Cref{thm:grundy-lower} and \Cref{prop:tcf-monotone},
$\chiStable(K_{m,n})\ge\chiStable(H)\ge\Gamma(H)\ge m+1$.
\end{proof}

The proofs also give polynomial-time algorithms that, given a preference
profile, construct stable colorings of $P_n$, $C_n$, and $K_{m,n}$ using
$\chiStable(P_n)$, $\chiStable(C_n)$, and $\chiStable(K_{m,n})$ colors,
respectively.

\subsection{Bound in Terms of Maximum Degree}\label{subsec:maxdegree}
A greedy coloring algorithm uses at most $\Delta+1$ colors on every graph $G$ of maximum degree $\Delta$, and hence $\chi(G)\le\Delta+1$. This bound does not extend to the stable chromatic number; in the appendix, we present an example where $\chiStable(G) > \Delta(G) + 1$.
Still, the stable chromatic number is bounded by a function of $\Delta(G)$.

\begin{theorem}\label{thm:max-degree}
For every graph $G$ of maximum degree $\Delta$,
$\chiStable(G)\le 2^\Delta$.
\end{theorem}

\begin{proof}
We construct an acyclic edge orientation with at most $2^\Delta$ vertices reachable from each vertex. For $i=1,\ldots,\Delta+1$, suppose that colors $1,\ldots,i-1$ have already been assigned. Let $X_i$ be the set of uncolored vertices with fewer than $i$ colored neighbors. Choose a maximal independent set $S_i$ of $G[X_i]$ and assign color $i$ to every vertex in $S_i$.

After phase $i$, every uncolored vertex has at least $i$ colored neighbors. We prove this by induction on $i$. The claim holds for $i=0$. Let $i\ge 1$ and consider a vertex $v$ that remains uncolored after phase $i$. If $v\notin X_i$, then it already had at least $i$ colored neighbors before this phase. If $v\in X_i$, then the induction hypothesis gives at least $i-1$ colored neighbors, and the maximality of $S_i$ gives a neighbor in $S_i$. This proves the claim. In particular, no vertex remains uncolored after phase $\Delta+1$.

Let $\psi$ be the resulting coloring, and orient every edge from the endpoint with larger $\psi$-value to the endpoint with smaller $\psi$-value. Each $S_i$ is independent (in $G[X_i]$, hence $G$), so $\psi$ is proper and the orientation is well-defined. It is acyclic because $\psi$ strictly decreases along every directed edge.

We prove by induction on $i$ that at most $2^{i-1}$ vertices are reachable from any vertex $v$ with $\psi(v)=i$, including $v$ itself.
The claim holds for $i=1$. Let $i\ge 2$. Immediately before phase $i$, the vertex $v$ has exactly $i-1$ colored neighbors: it has at least $i-1$ by the invariant and fewer than $i$ because $v\in X_i$. Denote these neighbors by $u_1,\ldots,u_{i-1}$ so that $\psi(u_1)\le\cdots\le\psi(u_{i-1})$. For every $j\in[i-1]$, the invariant after phase $j$ implies that $v$ has at least $j$ neighbors of color at most $j$. Hence $\psi(u_j)\le j$. Every directed path starting at $v$ first enters one of these neighbors. By the induction hypothesis, the number of vertices reachable from $v$ is therefore at most
\[
    1+\sum_{j=1}^{i-1}2^{\psi(u_j)-1}
    \le 1+\sum_{j=1}^{i-1}2^{j-1}
    =2^{i-1}.
\]
Thus, $\mathrm{reach}(G)\le 2^\Delta$, and \Cref{thm:reachability} gives $\chiStable(G)\le 2^\Delta$.
\end{proof}

\subsection{Bounds in Terms of Treewidth}\label{subsec:tw}

A \emph{tree decomposition} of $G$ consists of a tree $T$ and bags $B_x\subseteq V(G)$ indexed by $x\in V(T)$ such that every vertex belongs to a bag, the endpoints of every edge belong to a common bag, and the bags containing any fixed vertex induce a connected subtree of $T$. Its width is $\max_{x\in V(T)}|B_x|-1$.
The treewidth of $G$, denoted by $\tw(G)$, is the minimum width of a tree decomposition of $G$.

Every graph $G$ of treewidth $t$ satisfies $\chi(G)\le t+1$~\citep{cygan2015parameterized}. We show that its stable chromatic number is bounded in terms of $t$ and the number of vertices.
We use the following balanced-separator property of tree decompositions due to \citet[(2.5)--(2.6)]{robertson1986graph}.

\begin{lemma}[\citeauthor{robertson1986graph}~\citeyear{robertson1986graph}]\label{lem:treewidth-balanced-bag}
Let $G$ be an $n$-vertex graph with a tree decomposition of width $t$.
There is a bag $S$ of size at most $t+1$ such that every connected
component of $G-S$ has at most $n/2$ vertices.
\end{lemma}

\begin{theorem}\label{thm:treewidth-color-bound}
Let $G$ be an $n$-vertex graph, and suppose that a width-$t$ tree
decomposition of $G$ is given. Then
\[
    \chiStable(G)
    \le
    (t+1)\left(\left\lceil
        \log_2\frac{n}{t+1}\right\rceil+1\right).
\]
Given a preference profile, a stable coloring using at most this many
colors can be constructed in polynomial time.
\end{theorem}

\begin{proof}
Let $b=t+1$. For every positive integer $m$, define
\[
    q(m)=b\left(
        \max\left\{0,\left\lceil\log_2\frac{m}{b}\right\rceil\right\}
        +1
    \right).
\]
We prove by induction on $m$ that every $m$-vertex graph admitting a tree decomposition of width at most $t$ has an acyclic edge orientation in which at most $q(m)$ vertices are reachable from any vertex.
If $m\le b$, any acyclic orientation has this property because $q(m)=b\ge m$.

Let $m>b$, and let $F$ be such an $m$-vertex graph together with a tree decomposition of width at most $t$. Apply \Cref{lem:treewidth-balanced-bag} and let $S$ be the resulting bag.
Every component $H$ of $F-S$ has at most $m/2$ vertices and inherits a
tree decomposition of width at most $t$ by restricting the given bags
to $V(H)$. Write $\ell=\lceil\log_2(m/b)\rceil$. If $|V(H)|\le b$, then
$q(|V(H)|)=b\le b\ell$.
Otherwise,
\[
    \left\lceil\log_2\frac{|V(H)|}{b}\right\rceil
    \le
    \left\lceil\log_2\frac{m}{b}\right\rceil-1
    =\ell-1,
\]
so again $q(|V(H)|)\le b\ell=q(m)-b$.

Orient each component as given by the induction hypothesis, orient every
edge between $F-S$ and $S$ toward $S$, and orient the edges of $F[S]$
acyclically. Each component of $F-S$ and the graph $F[S]$ are oriented
acyclically, and no edge is directed from $S$ to $F-S$. Hence, the
resulting orientation is acyclic. A vertex in $S$ reaches only vertices
in $S$, so it reaches at most $|S|\le b\le q(m)$ vertices. A vertex in
a component $H$ reaches only vertices in $H\cup S$, so it reaches at
most $q(|V(H)|)+|S|\le q(m)$ vertices. This proves the induction claim.

Apply the claim to $G$ with $m=n$. Since $b\le n$, the value $q(n)$
equals the bound in the statement, and \Cref{thm:reachability} gives
the stated bound. At each recursive call, a suitable bag can be found
by checking the bags of the given decomposition. Hence, the orientation
and the corresponding stable coloring can be constructed in polynomial
time.
\end{proof}

We complement the upper bound with the following lower bound.
\begin{proposition}\label{prop:treewidth-tight}
For all integers $s\ge 1$ and $r\ge 2$, there is a graph $G_{s,r}$ on
$n=s2^{r-1}$ vertices with $\tw(G_{s,r})=2s-1$ and
$sr\le\chiStable(G_{s,r})\le 2s(r-1)$.
\end{proposition}

\begin{proof}
Let $B_1$ be the one-vertex rooted tree. For $i\ge 2$, take one disjoint
copy of each of $B_1,\ldots,B_{i-1}$, add a new root, and join it to the
root of every copy. Then $B_r$ is a tree on $2^{r-1}$ vertices.
Inductively, it has a Grundy coloring $\gamma$ with
$\gamma(V(B_r))=[r]$ in which its root has color $r$.

Define $G_{s,r}$ by replacing each vertex $v$ of $B_r$ by a clique
$X_v$ of size $s$ and adding all edges between $X_u$ and $X_v$ whenever
$uv\in E(B_r)$. Assign the colors
$(\gamma(v)-1)s+1,\ldots,\gamma(v)s$ bijectively to $X_v$. Let
$\gamma(v)=i$ and $a\in[s]$, and consider the vertex in $X_v$ of color
$(i-1)s+a$. Within $X_v$, it has a neighbor of
each smaller color in the same block. For every $j<i$, the Grundy
property gives a neighbor $u$ of $v$ with $\gamma(u)=j$, and $X_u$
contains all $s$ colors in the $j$th block. The resulting coloring is
therefore a Grundy coloring with $sr$ colors, so
\Cref{thm:grundy-lower} gives $\chiStable(G_{s,r})\ge sr$.

At the root $v$ of $B_r$, use the bag $X_v$. For every non-root vertex
$v$ with parent $u$, use the bag $X_v\cup X_u$. These bags, indexed by
$V(B_r)$, form a tree decomposition of width $2s-1$.
Conversely, $G_{s,r}$ contains a clique of size $2s$ corresponding to
every edge of $B_r$, so its treewidth is exactly $2s-1$.
The upper bound follows from \Cref{thm:treewidth-color-bound}, since
$n/(2s)=2^{r-2}$.
\end{proof}

The ratio between the upper and lower bounds in
\Cref{prop:treewidth-tight} is $2-2/r<2$. This shows that the bound is asymptotically tight for trees.

\begin{corollary}\label{cor:treewidth-tree-lower}
For a tree $T$ on $n$ vertices, $\chiStable(T) = O( \log n)$. Furthermore, for each $n$, there exists a tree $T$ of $n$ vertices such that $\chiStable(T)=\Omega(\log n)$.
\end{corollary}
\begin{proof}
    For $n=2^{r-1}$, consider the tree $B_r$ above.
    For general $n$, take $r$ such that $2^{r-1} \le n< 2^r$ and attach $n-2^{r-1}$ new leaves to the root of $B_r$ to make a tree $T$ of $n$ vertices. By \Cref{prop:tcf-monotone}, $\chiStable(T) \ge \chiStable(B_r) = \Omega( \log n)$.
\end{proof}

\subsection{Planar Graphs}\label{subsec:planar}

By the Four Color Theorem~\citep{robertsonEtAl1997four}, every planar graph $G$ satisfies $\chi(G)\le 4$.
However, since trees are planar, \Cref{cor:treewidth-tree-lower} gives an $\Omega(\log n)$ lower bound for planar graphs.

Recursively applying the planar separator theorem \citep{lipton1979separator} shows that every $n$-vertex planar graph has treewidth $O(\sqrt n)$. Thus, \Cref{thm:treewidth-color-bound} gives $\chiStable(G)=O(\sqrt{n}\log n)$.
In fact, we show that $\chiStable(G)$ is polylogarithmic in the number of vertices.

\begin{theorem}\label{thm:planar-color-bound}
For every planar graph $G$ on $n$ vertices,
$\chiStable(G)=O(\log^2 n)$.
\end{theorem}
\begin{proof}
If $n=1$, then one color suffices, so assume that $n\ge 2$.
For graphs $A$ and $B$, the strong product $A\boxtimes B$ has vertex set $V(A)\times V(B)$. Two distinct vertices $(a,b)$ and $(a',b')$ are adjacent if and only if $a=a'$ or $aa'\in E(A)$, and $b=b'$ or $bb'\in E(B)$.

By the planar product structure theorem~\citep{dujmovicEtAl2020queue}, $G$ is a subgraph of $P\boxtimes H$ for a path $P$ and a graph $H$ of treewidth at most $8$. Identify $G$ with its image in this product, and
let $X$ be the set of second ($H$) coordinates used by its vertices. Replacing $H$ by $H[X]$ preserves the inclusion $G\subseteq P\boxtimes H$ and ensures that $|V(H)|\le n$ and $\tw(H)\le 8$.

Apply the construction in the proof of \Cref{thm:treewidth-color-bound} to $H$. It gives an acyclic orientation $D_H$ in which every vertex reaches at most $r$ vertices, including itself, for some $r=O(\log n)$. Write $P=p_1\ldots p_\ell$. We define an orientation $D$ of $P\boxtimes H$. For an edge $(p_i,h)(p_j,h')$ with $h\ne h'$, orient it from $(p_i,h)$ to $(p_j,h')$ exactly when $D_H$ orients $hh'$ from $h$ to $h'$. For an edge with $h=h'$, orient it from the endpoint with an even-indexed first ($P$) coordinate to the endpoint with an odd-indexed first ($P$) coordinate.

We show that $D$ is acyclic and that the vertices reachable from any fixed vertex use $O(\log n)$ distinct values in each coordinate.

Fix a topological ordering $\pi\colon V(H)\to[|V(H)|]$ of $D_H$. If an arc of $D$ goes from $(p_i,h)$ to $(p_j,h')$ with $h\ne h'$, then $\pi(h)<\pi(h')$. Hence, no directed cycle in $D$ can change its second coordinate. Within each fixed second coordinate, every arc is directed from an even-indexed first coordinate to an odd-indexed one, so no directed cycle exists.
Thus, $D$ is acyclic.

Consider a directed path in $D$ starting at $(p_i,h)$.
We claim that the second coordinate changes at most $r-1$ times.
Deleting consecutive repetitions from its sequence of second coordinates, we obtain a directed walk in $D_H$. Since $D_H$ is acyclic, no vertex occurs twice in this walk, and hence it is a directed path.
The second coordinate therefore changes no more than $r-1$ times.

Each maximal subpath with a fixed second coordinate contains at most one edge because the alternating orientation of $P$ has no directed path of length two. There are at most $r$ such subpaths, so the original path has length at most $(r-1)+r=2r-1$. Every edge of the strong product changes the first coordinate by distance at most one in $P$. Thus, the final first coordinate is at distance at most $2r-1$ from $p_i$.

The second coordinate of every vertex reachable from $(p_i,h)$ is one of the at most $r$ vertices reachable from $h$ in $D_H$. Its first coordinate is one of the at most $4r-1$ vertices of $P$ at distance at most $2r-1$ from $p_i$. Hence, at most $r(4r-1)=O(\log^2 n)$ vertices are reachable from any vertex of $P\boxtimes H$. Hence, by \Cref{thm:reachability} and \Cref{prop:tcf-monotone}, $\chiStable(G) \le \chiStable(P\boxtimes H)=O(\log^2 n)$.
\end{proof}

The $\Omega(\log n)$ lower bound leaves a gap of one logarithmic factor.

\section{Algorithms and Complexity}\label{sec:algorithms}

In this section, we consider the complexity of finding a stable coloring with the minimum number of colors for a given instance $(G, \succ)$.

\subsection{Complexity for a Fixed Number of Colors}

We first prove a complexity dichotomy for every fixed number of colors. For an integer $k$, let \emph{stable $k$-colorability} to be the decision problem that, whether a given instance admits a stable $k$-coloring.

\begin{theorem}\label{thm:general-complexities}
Stable $2$-colorability is solvable in polynomial time. In this case, a stable coloring can also be constructed in polynomial time. 
For every fixed integer $k\ge 3$, stable $k$-colorability is NP-complete.
\end{theorem}
\begin{proof}
We first consider stable 2-colorability.
If $G$ is not bipartite, then it has no proper $2$-coloring and hence no stable $2$-coloring.
Assume that $G$ is bipartite. For each connected component $C$ of $G$, fix a bipartition $(L_C,R_C)$. Up to swapping the two colors, there are only two proper $2$-colorings of $C$: one assigning color $1$ to $L_C$ and color $2$ to $R_C$, and the other assigning color $2$ to $L_C$ and color $1$ to $R_C$.

For each component, test both colorings by constructing their envy graphs and checking acyclicity. If neither envy graph is acyclic, reject. Otherwise, choose a coloring with an acyclic envy graph. Since there are no edges between distinct components, the union of the chosen colorings is stable. The algorithm runs in polynomial time.

Next, fix $k \ge 3$. We will consider stable $k$-colorability.
The problem is in NP because we can check properness, construct the envy digraph, and test acyclicity in polynomial time.

For hardness, we reduce from ordinary $k$-colorability, which is NP-complete for every fixed $k\ge3$ \citep{garey1979computers}.
Given a graph $G$, give every vertex the same ranking on $[k]$,
\[
    1\succ 2\succ \cdots\succ k,
\]
and extend the ranking beyond $[k]$ arbitrarily.
We claim that $G$ has a $k$-coloring if and only if $(G,\succ)$ has a stable $k$-coloring. Let $\phi$ be a $k$-coloring of $G$. Every arc $(u,v)$ in the envy graph satisfies $\phi(v)<\phi(u)$. Thus, a directed cycle $v_1v_2\ldots v_tv_1$ would imply
\[
    \phi(v_1)>\phi(v_2)>\cdots>\phi(v_t)>\phi(v_1),
\]
a contradiction. Hence, $\phi$ is stable. Conversely, every stable $k$-coloring of $(G,\succ)$ is a proper $k$-coloring of $G$ by definition.
\end{proof}

\subsection{FPT Algorithm Parameterized by Treewidth}
Given the computational hardness in general, we now design an FPT algorithm that finds a stable coloring using the minimum possible number of colors, parameterized by treewidth.
Recall that an algorithm is FPT with respect to a parameter $t$ if it runs in time $f(t)N^{O(1)}$ on instances of size $N$, for some computable function $f$.

For an $n$-vertex graph $G$ of treewidth $t$, the 2-approximation algorithm of Korhonen computes a tree decomposition of width at most $2t+1$ in $2^{O(t)}n$ time \citep{korhonen2021single}.

We use a rooted nice tree decomposition with empty root and leaf bags and four types of internal nodes: introduce-vertex, introduce-edge, forget-vertex, and join nodes.
An introduce-vertex node $x$ with child $y$ satisfies $B_x=B_y\cup\{v\}$. A forget-vertex node $x$ with child $y$ satisfies $B_x=B_y\setminus\{v\}$. An introduce-edge node $x$ with child $y$ satisfies $B_x=B_y$ and introduces one edge whose endpoints both belong to the bag.
A join node $x$ has two children $y,z$ with $B_x=B_y=B_z$.
Every edge is introduced exactly once, before either endpoint is forgotten.
A width-$w$ tree decomposition can be converted in polynomial time into a nice tree decomposition of the same width with $O((w+1)n)$ nodes \citep{cygan2015parameterized}.

We first give a dynamic program parameterized by the decomposition width and the number of colors. For a graph of treewidth $t\ge 1$, \Cref{thm:treewidth-color-bound} shows that $O(t\log n)$ colors suffice. Together with Korhonen's 2-approximation, this establishes fixed-parameter tractability with respect to treewidth.

\begin{proposition}\label{prop:treewidth-dp}
Given a width-$t$ tree decomposition of $G$, one can decide whether a stable coloring instance on $G$ admits a stable $k$-coloring in time $k^{t+1}2^{O(t^2)}n^{O(1)}$. If the answer is positive, the algorithm also constructs such a coloring.
\end{proposition}

\begin{proof}
If $k\ge n$, the construction in the proof of \Cref{thm:reachability} gives a stable coloring using colors from $[n]\subseteq[k]$ in polynomial time. Hence, assume that $k<n$.
Fix a nice tree decomposition.
For a node $x$, let $V_x$ be the union of all bags in the subtree rooted at $x$, let $E_x$ be the set of edges introduced in this subtree, and write $G_x=(V_x,E_x)$.
For a coloring $\varphi$ of $G_x$, let $D_x^\varphi$ be its envy graph with respect to the edges in $E_x$.

For each node $x$, we compute a Boolean table $T_x$ indexed by pairs $(\alpha,R)$, where $\alpha\colon B_x\to [k]$ and $R\subseteq B_x\times B_x$.
Here and below, $(v,v)$ is not included in a reachability relation $R$ unless the digraph contains a directed cycle through $v$.
We set $T_x[\alpha,R]$ to true if and only if there is a proper coloring $\varphi\colon V_x\to[k]$ extending $\alpha$ such that the digraph $D_x^\varphi$ is acyclic and, for every $u,v\in B_x$, the pair $(u,v)$ belongs to $R$ exactly when $D_x^\varphi$ contains a directed path from $u$ to $v$.

We process the nodes bottom-up.
A leaf table has one true entry, indexed by the empty assignment and the empty relation.
At an introduce-vertex node for $v$, for every true child entry, we extend its assignment by each color $c \in [k]$ for $v$ without adding any ordered pair involving $v$ to the relation.
At an introduce-edge node for $uv$, we discard every child entry with $\alpha(u)=\alpha(v)$.
We then add the envy arcs determined by the preferences of $u$ and $v$:
\[
    A_{uv}
    =
    \{(u,v):\alpha(v)\succ_u\alpha(u)\}
    \cup
    \{(v,u):\alpha(u)\succ_v\alpha(v)\}.
\]
Let $R'$ be the reachability relation of the directed graph $(B_x,R\cup A_{uv})$. That is, for $a,b\in B_x$, we have $(a,b)\in R'$ if and only if $a$ reaches $b$ in this graph.
If $(a,a)\notin R'$ for every $a\in B_x$, we set $T_x[\alpha,R']$ to true.
At a forget-vertex node for $v$, for every true child entry, we set the entry obtained by restricting both $\alpha$ and $R$ to the remaining bag to true.
At a join node $x$ with children $y$ and $z$, we consider pairs of true entries $T_y[\alpha,R_1]$ and $T_z[\alpha,R_2]$ with the same assignment $\alpha$. For every such pair, let $R$ be the reachability relation of the directed graph $(B_x,R_1\cup R_2)$. If $(a,a)\notin R$ for every $a\in B_x$, we set $T_x[\alpha,R]$ to true.

We prove by induction that the true entries of each table are exactly those satisfying the invariants above.
A leaf has an empty bag and no introduced edges, so the empty coloring witnesses its unique true entry.
At an introduce-vertex node, the new vertex is incident with no introduced edge. Assigning it any color therefore creates no envy arc and leaves the reachability relation unchanged.
At an introduce-edge node, $A_{uv}$ is exactly the set of arcs added to the child envy graph. Since $R$ records all child reachability between bag vertices, $R'$ is precisely the reachability relation induced on $B_x$ by the resulting envy graph. The child envy graph is acyclic, and every new arc has both endpoints in $B_x$. Hence, the resulting envy graph contains a directed cycle exactly when $(a,a)\in R'$ for some $a\in B_x$.
At a forget-vertex node, restricting $R$ to the new bag preserves reachability between its vertices, including reachability witnessed by a path containing the forgotten vertex as an internal vertex.
At a join node, the two child envy graphs intersect exactly in the common bag and have disjoint arc sets (since edges are introduced only once). Every directed path in their union whose endpoints lie in $B_x$ decomposes into subpaths that alternate between the two child envy graphs, with the endpoints of each subpath lying in the common bag. Thus, the endpoints of the original path belong to the reachability relation of $(B_x,R_1\cup R_2)$. Conversely, concatenating the corresponding child paths realizes every pair in this relation. Therefore, every transition sets precisely the feasible table entries to true.
The entry $T_r[\emptyset,\emptyset]$ at the root $r$ is true if and only if the full instance admits a stable $k$-coloring.

Each bag has at most $t+1$ vertices. Thus, each table has at most $k^{t+1}2^{(t+1)^2}$ entries. Every non-join transition takes polynomial time per entry. At a join node, for each assignment $\alpha$, we consider at most $2^{2(t+1)^2}$ pairs $(R_1,R_2)$, and the number of $\alpha$ to consider is at most $k^{t+1}$. Since the nice decomposition has $O((t+1)n)$ nodes, the total running time is $k^{t+1}2^{O(t^2)}n^{O(1)}$. Keeping one predecessor for every true entry allows us to reconstruct a coloring.
\end{proof}

\begin{theorem}\label{thm:treewidth-fpt}
For every stable coloring instance on an $n$-vertex graph $G$ of
treewidth $t$, a stable coloring using the minimum possible number of
colors can be found in time $2^{O(t^2)}n^{O(1)}$. Consequently, stable
$k$-colorability is fixed-parameter tractable with respect to $t$.
\end{theorem}

\begin{proof}
If $t=0$, then $G$ is edgeless, so one color is optimal and can be found in polynomial time.
Suppose that $t\ge 1$.
A tree decomposition of width $w\le 2t+1$ can be found in $2^{O(t)}n$ time \citep{korhonen2021single} and converted into the nice form used in \Cref{prop:treewidth-dp}.
Set
\[
    q=\min\left\{
        n,\,
        (w+1)\left(\left\lceil
        \log_2\frac{n}{w+1}\right\rceil+1\right)
    \right\}.
\]
By \Cref{cor:finite} and \Cref{thm:treewidth-color-bound}, every preference profile admits a stable $q$-coloring.
We run the algorithm of \Cref{prop:treewidth-dp} for $k=1,\ldots,q$ in increasing order and return the first coloring found, which therefore uses the minimum possible number of colors.

For each $k\le q$, \Cref{prop:treewidth-dp} takes time at most $q^{w+1}2^{O(w^2)}n^{O(1)}$. We have $q\le (w+1)(2+\log_2 n)$. For $r=w+1$ and $n\ge 2$, the inequality $(2+\log_2 n)^r\le (3r)^r n$ gives $q^{w+1}\le (3r^2)^r n=2^{O(w\log(w+2))}n$. Since $w\le 2t+1$, each call takes $2^{O(t^2)}n^{O(1)}$ time.
There are at most $q\le n$ calls, so the total running time remains $2^{O(t^2)}n^{O(1)}$.
\end{proof}

\begin{corollary}
    For an instance $(T, \succ)$ where $T$ is a tree, a stable coloring using the minimum possible number of colors can be found in polynomial time.
\end{corollary}

\section{Conclusion}
We introduced stable graph coloring in which vertices possess preferences over the possible colors. We established bounds for several graph classes, and gave an FPT algorithm that finds a stable coloring with the minimum possible number of colors, parameterized by treewidth. Determining the stable chromatic number of planar graphs remains open. Studying other graph classes, such as interval graphs, will also be fruitful.

\paragraph{Acknowledgments}
\texttt{ChatGPT-5.6-Sol} established the connection between the planar product structure theorem and reachability bound, and found the instance where $\chiStable(G) > \Delta(G) + 1$.

\bibliography{stable_coloring}

\appendix
\section{An Instance where $\chiStable(G) > \Delta + 1$}

For graphs $A$ and $B$, the \emph{Cartesian product} $A \square B$ has vertex set $V(A)\times V(B)$. Two distinct vertices $(a,b)$ and $(a',b')$ are adjacent if and only if [$a=a'$ and $bb'\in E(B)$] or [$aa'\in E(A)$ and $b=b'$].

\begin{proposition}
Let $P=C_3\square K_2$ be the triangular prism graph. Then
\[
    \chi_{\mathrm{stable}}(P)>\Delta(P)+1.
\]
\end{proposition}

\begin{proof}
Let the two triangles of $P$ be
\[
    abc
    \qquad\text{and}\qquad
    a'b'c',
\]
with matching edges $aa'$, $bb'$, and $cc'$. Since $\Delta(P)=3$, it suffices to prove that $\chi_{\mathrm{stable}}(P) \ge 5$.

We construct a preference profile for which no stable $4$-coloring exists. Give corresponding vertices in the two triangles the same preferences:
\[
\begin{aligned}
    1 &\succ_a 2 \succ_a 3 \succ_a 4,
    &\qquad
    1 &\succ_{a'} 2 \succ_{a'} 3 \succ_{a'} 4,\\
    3 &\succ_b 2 \succ_b 1 \succ_b 4,
    &
    3 &\succ_{b'} 2 \succ_{b'} 1 \succ_{b'} 4,\\
    4 &\succ_c 2 \succ_c 1 \succ_c 3,
    &
    4 &\succ_{c'} 2 \succ_{c'} 1 \succ_{c'} 3.
\end{aligned}
\]

We first determine the possible restriction of a stable $4$-coloring to one of the triangles. Let $S\subseteq[4]$ be the set of three colors used on $abc$. For each possible set $S$, the favorite colors of $a,b,c$ within $S$ are as follows:
\[
\begin{array}{c|c}
S
&
\bigl(
    \operatorname{top}_a(S),
    \operatorname{top}_b(S),
    \operatorname{top}_c(S)
\bigr)
\\ \hline
\{1,2,3\} & (1,3,2)\\
\{1,2,4\} & (1,2,4)\\
\{1,3,4\} & (1,3,4)\\
\{2,3,4\} & (2,3,4)
\end{array}
\tag{$\ast$}
\]

In every row, the three favorite colors are distinct. We claim that a stable coloring of the triangle using the color set $S$ must assign every vertex its favorite color within $S$.

Indeed, a coloring in which each vertex receives its favorite color on the triangle is stable. Let $\phi$ be a proper coloring of the triangle such that not all vertices receive its favorite color. Let $S$ be the three colors used and $v$ be one vertex that $\phi(v) \ne \operatorname{top}_v(S)$. Then, $(\ast)$ shows that one of the neighbors of $v$ has $\phi(v)$ as its favorite color. Call it $w$. If $\phi(w) = \operatorname{top}_v(S)$, then we have a blocking cycle $vwv$. If $\phi(w) \ne \operatorname{top}_v(S)$, then for the remaining vertex $u \in \{a,b,c\} \setminus \{v,w\}$, it must be that $\phi(u) = \operatorname{top}_v(S)$. Since the favorite colors are all distinct, $\phi(u) \ne \operatorname{top}_u(S)$, and hence there is an envy arc from $u$ to either $v$ or $w$. In either case, we have a blocking cycle $vuv$ or $vuwv$, respectively. Hence, the claim holds.

Consequently, the restriction of any stable $4$-coloring to $abc$ must be one of the four color vectors in~$(\ast)$. The same conclusion holds for $a'b'c'$, because corresponding vertices have identical preferences.

However, any two vectors in~$(\ast)$ agree in at least one coordinate. Therefore, the color vectors assigned to $(a,b,c)$ and $(a',b',c')$ agree at $a,a'$, at $b,b'$, or at $c,c'$, contradicting properness.
Hence this preference profile admits no stable $4$-coloring, and therefore $\chi_{\mathrm{stable}}(P)\geq 5$.
\end{proof}

\end{document}